\documentclass{tlp}

\usepackage{graphicx}
\usepackage{amsmath}
\usepackage{amssymb}
\usepackage{amsthm}
\usepackage{algorithm}
\usepackage[noend]{algpseudocode}

\usepackage{hyperref}

\newtheorem{definition}{Definition}
\newtheorem{lemma}{Lemma}
\newtheorem{example}{Example}
\newtheorem{proposition}{Proposition}

\newcommand{\coloneq}{\mbox{ :- }}
\newcommand{\val}{\simeq\!\!}

\DeclareSymbolFont{bbold}{U}{bbold}{m}{n}
\DeclareSymbolFontAlphabet{\mymathbb}{bbold}

\begin{document} 

\title{The Magic of Logical Inference in Probabilistic Programming}

\author[B.~Gutmann, I.~Thon, A.~Kimmig, M.~Bruynooghe, L.~De~Raedt]{Bernd Gutmann,
            Ingo Thon, Angelika Kimmig, Maurice Bruynooghe and Luc De~Raedt \\
            Department of Computer Science, Katholieke Universiteit Leuven,\\
            Celestijnenlaan 200A - bus 2402, 3001 Heverlee, Belgium\\
            \{firstname.lastname\}@cs.kuleuven.be}

 \maketitle

\begin{abstract} 
  Today, many different probabilistic programming languages exist and
  even more inference mechanisms for these languages.  Still, most
  logic programming based languages use backward reasoning based on
  SLD resolution for inference. While these methods are typically
  computationally efficient, they often can neither handle infinite
  and/or continuous distributions, nor evidence.  To overcome these
  limitations, we introduce distributional clauses, a variation and
  extension of Sato's distribution semantics. We also contribute a
  novel approximate inference method that integrates forward reasoning
  with importance sampling, a well-known technique for probabilistic
  inference. To achieve efficiency, we integrate two logic programming
  techniques to direct forward sampling. Magic sets are used to focus
  on relevant parts of the program, while the integration of backward
  reasoning allows one to identify and avoid regions of the sample
  space that are inconsistent with the evidence.
\end{abstract}

\section{Introduction} 
The advent of statistical relational learning
\cite{Getoor07,DeRaedtAPRIL08} and probabilistic programming
\cite{NIPSWorkshop} has resulted in a vast number of different languages and systems such as
PRISM~\cite{SatoKameya:01}, ICL~\cite{Poole08},
ProbLog~\cite{DeRaedt07-IJCAIa}, Dyna~\cite{Eisner05},
BLPs~\cite{Kersting08}, CLP($\mathcal{BN}$)~\cite{clpbn},
BLOG~\cite{Milch05}, Church~\cite{Goodman08}, IBAL~\cite{Pfeffer01},
and MLNs~\cite{Richardson:06}. While inference in these languages
generally involves evaluating the probability distribution defined by
the model, often conditioned on evidence in the form of known truth
values for some atoms, this diversity of systems has led to a variety
of inference approaches.  Languages such as IBAL, BLPs, MLNs and
CLP($\mathcal{BN}$) combine knowledge-based model construction to
generate a graphical model with standard inference techniques for such
models.  Some probabilistic programming languages, for instance BLOG
and Church, use sampling for approximate inference in generative
models, that is, they estimate probabilities from a large number of
randomly generated program traces. Finally, probabilistic logic
programming frameworks such as ICL, PRISM and ProbLog, combine
SLD-resolution with probability calculations.

So far, the second approach based on sampling has received little
attention in logic programming based systems.  In this paper, we
investigate the integration of sampling-based approaches into
probabilistic logic programming frameworks to broaden the
applicability of these.  Particularly relevant in this regard are the
ability of Church and BLOG to sample from continuous distributions and
to answer conditional queries of the form $p(q |e )$ where $e$ is the
evidence.  To accommodate (continuous and discrete) distributions, we
introduce \emph{distributional clauses}, which define random variables
together with their associated distributions, conditional upon logical
predicates.  Random variables can be passed around in the logic
program and the outcome of a random variable can be compared with
other values by means of special built-ins.  To formally establish the
semantics of this new construct, we show that these random variables
define a basic distribution over facts (using the comparison
built-ins) as required in Sato's distribution
semantics~\cite{Sato:95}, and thus induces a distribution over least
Herbrand models of the program. This contrasts with previous instances
of the distribution semantics in that we no longer enumerate the
probabilities of alternatives, but instead use arbitrary densities and
distributions. 

From a logic programming perspective, BLOG~\cite{milch:aistats2005}
and related approaches perform \emph{forward reasoning}, that is, the
samples needed for probability estimation are generated starting from
known facts and deriving additional facts, thus generating a
\emph{possible world}. PRISM and related approaches follow the
opposite approach of \emph{backward reasoning}, where inference starts
from a query and follows a chain of rules backwards to the basic
facts, thus generating \emph{proofs}. This difference is one of the
reasons for using sampling in the first approach: exact forward
inference would require that all possible worlds be generated, which
is infeasible in most cases.  Based on this observation, we contribute
a new inference method for probabilistic logic programming that
combines sampling-based inference techniques with forward
reasoning. On the probabilistic side, the approach uses rejection
sampling~\cite{KollerFriedman09}, a well-known sampling technique that
rejects samples that are inconsistent with the evidence.  On the logic
programming side, we adapt the \emph{magic set}
technique~\cite{bancilhon} towards the probabilistic setting, thereby
combining the advantages of forward and backward
reasoning. Furthermore, the inference algorithm is improved along the
lines of the \emph{SampleSearch}
algorithm~\cite{Gogate09samplesearch:importance}, which avoids choices
leading to a sample that cannot be used in the probability estimation
due to inconsistency with the evidence. We realize this using a
heuristic based on backward reasoning with limited proof length, the
benefit of which is experimentally confirmed.  This novel approach to
inference creates a number of new possibilities for applications of
probabilistic logic programming systems, including continuous
distributions and Bayesian inference.

This paper is organized as follows: we start by reviewing the basic
concepts in Section~\ref{sec:prelim}.  Section~\ref{sec:semantics}
introduces the new language and its semantics,
Section~\ref{sec:algorithms} a novel forward sampling algorithm for
probabilistic logic programs. Before concluding, we evaluate our
approach in Section~\ref{sec:experiments}.

\section{Preliminaries}
\label{sec:prelim}
\subsection{Probabilistic Inference} 
\label{sec:probinf}

A discrete probabilistic model defines a probability
distribution~$p(\cdot)$ over a set~$\Omega$ of basic outcomes, that
is, value assignments to the model's random variables. This
distribution can then be used to evaluate a conditional probability
distribution $p(q|e) =\frac{p(q\wedge e)}{p(e)}$, also called
\emph{target distribution}.  Here, $q$ is a query involving random
variables, and $e$ is the \emph{evidence}, that is, a partial value
assignment of the random variables\footnote{If $e$ contains
  assignments to continuous variables, $p(e)$ is zero. Hence, evidence
  on continuous values has to be defined via a probability density
  function, also called a sensor model.}.  Evaluating this target
distribution is called \emph{probabilistic
  inference}~\cite{KollerFriedman09}. In probabilistic logic
programming, random variables often correspond to ground atoms,
and~$p(\cdot)$ thus defines a distribution over truth value
assignments, as we will see in more detail in Sec.~\ref{sec:ds} (but
see also~\citeNP{NIPSWorkshop}). Probabilistic inference then asks for
the probability of a logical query being true given truth value
assignments for a number of such ground atoms.

In general, the probability~$p(\cdot)$ of a query~$q$ is in the
discrete case the sum over those outcomes $\omega\in \Omega$ that are
consistent with the query.  In the continuous case, the sum is
replaced by an (multidimensional) integral and the
distribution~$p(\cdot)$ by a (product of) densities
$\mathbf{F}(\cdot)$ That is,
\begin{equation}\label{eq:exact}
p(q)= \sum_{\omega\in \Omega} p(\omega) \mymathbb{1}_{q}(\omega), \quad \quad and\quad \quad p(q)=\idotsint\limits_\Omega \mymathbb{1}_q(\omega) d\mathbf{F}(\omega)
\end{equation}
where $\mymathbb{1}_{q}(\omega)=1$ if $\omega\models q$ and $0$
otherwise.  As common (e.g.~\cite{wasserman04allof}) we will use for
convenience the notation $\int xdF(x)$ as unifying notation for both
discrete and continuous distributions.

As $\Omega$ is often very large or even infinite, exact inference
based on the summation in~\eqref{eq:exact} quickly becomes infeasible,
and inference has to resort to approximation techniques based on
\emph{samples}, that is, randomly drawn outcomes $\omega\in
\Omega$. Given a large set of such samples $\{s_1,\ldots,s_N\}$ drawn
from~$p(\cdot)$, the probability~$p(q)$ can be estimated as the
fraction of samples where $q$ is true. If samples are instead drawn
from the target distribution $p(\cdot|e)$, the latter can directly be
estimated as
\begin{equation*}
  \hat{p}(q|e) := \frac{1}{N}\sum_{i=1}^N \mymathbb{1}_q(s_i) \enspace.
\end{equation*}

However, sampling from~$p(\cdot|e)$ is often highly inefficient or
infeasible in practice, as the evidence needs to be taken into
account. For instance, if one would use the standard definition of
conditional probability to generate samples from $p(\cdot)$, all
samples that are not consistent with the evidence do not contribute to
the estimate and would thus have to be discarded or, in sampling
terminology, \emph{rejected}.

More advanced sampling methods therefore often resort to a so-called
\emph{proposal distribution} which allows for easier sampling. The
error introduced by this simplification then needs to be accounted for
when generating the estimate from the set of samples.  An example for
such a method is \emph{importance sampling}, where each sample~$s_i$
has an associated \emph{weight}~$w_i$. Samples are drawn from an
\emph{importance distribution}~$\pi(\cdot |e)$, and weights are
defined as $w_i=\frac{p(s_i|e)}{\pi(s_i|e)}$. The true target
distribution can then be estimated as
\begin{equation*}
  \hat{p}(q|e) = \frac{1}{W} \sum_{i=1}^Nw_i\cdot \mymathbb{1}_q(s_i)  
\end{equation*}
where $W=\sum_i w_i$ is a normalization constant. The simplest
instance of this algorithm is \emph{rejection sampling} as already
sketched above, where the samples are drawn from the prior
distribution~$p(\cdot)$ and weights are $1$ for those samples
consistent with the evidence, and $0$ for the others. Especially for
evidence with low probability, rejection sampling suffers from a very
high rejection rate, that is, many samples are generated, but do not
contribute to the final estimate. This is known as the \emph{rejection
  problem}. One way to address this problem is \emph{likelihood
  weighted sampling}, which dynamically adapts the proposal
distribution during sampling to avoid choosing values for random
variables that cause the sample to become inconsistent with the
evidence. Again, this requires corresponding modifications of the
associated weights in order to produce correct estimates.

\subsection{Logical Inference} 
\label{sec:forward}

A (definite) clause is an expression of the form $\mathtt{h\coloneq
  b_1, \ldots ,b_n}$, where $\mathtt{h}$ is called head and
$\mathtt{b_1,\ldots,b_n}$ is the body.  A program consists of a set of
clauses and its semantics is given by its least Herbrand model.  There
are at least two ways of using a definite clause in a logical
derivation.  First, there is \emph{backward chaining}, which states
that to prove a goal $\mathtt{h}$ with the clause it suffices to prove
$\mathtt{b_1, \ldots,b_n}$; second, there is \emph{forward chaining},
which starts from a set of known facts $\mathtt{b_1, \ldots ,b_n}$ and
the clause and concludes that $\mathtt{h}$ also
holds~(cf.~\cite{nilsson:book}).  Prolog employs backward chaining
(SLD-resolution) to answer queries.  SLD-resolution is very efficient
both in terms of time and space. However, similar subgoals may be
derived multiple times if the query contains recursive
calls. Moreover, SLD-resolution is not guaranteed to always terminate
(when searching depth-first).  Using forward reasoning, on the other
hand, one starts with what is known and employs the immediate
consequence operator $T_P$ until a fixpoint is reached. This fixpoint
is identical to the least Herbrand model.

\begin{definition}[$T_P$ operator]
  \label{def:tpop}
  Let $P$ be a logic program containing a set of definite clauses and
  $ground(P)$ the set of all ground instances of these clauses.
  Starting from a set of ground facts $S$ the $T_P$ operator returns
  \begin{equation*}
    T_P(S) = \{ \mathtt{h} \mid \mathtt{h \coloneq b_1,\ldots,b_n} \in
    ground(P) \mbox{ and } \{ \mathtt{b_1}, \ldots , \mathtt{b_n}\}
    \subseteq S\}
  \end{equation*}
\end{definition}

\subsection{Distribution Semantics}
\label{sec:ds}

Sato's distribution semantics~\cite{Sato:95} extends logic programming
to the probabilistic setting by choosing truth values of basic facts
randomly.  The core of this semantics lies in splitting the logic
program into a set $F$ of \emph{facts} and a set $R$ of
\emph{rules}. Given a probability distribution $P_F$ over the facts,
the rules then allow one to extend $P_F$ into a distribution over
least Herbrand models of the logic program. Such a Herbrand model is
called a \emph{possible world}.

More precisely, it is assumed that $DB = F\cup R$ is ground and
denumerable, and that no atom in $F$ unifies with the head of a rule
in $R$. Each truth value assignment to $F$ gives rise to a unique
least Herbrand model of $DB$. Thus, a probability distribution $P_F$
over $F$ can directly be extended into a distribution $P_{DB}$ over
these models. Furthermore, Sato shows that, given an enumeration
$f_1,f_2,\ldots$ of facts in $F$, $P_F$ can be constructed from a
series of finite distributions $P_F^{(n)}(f_1=x_1,\ldots, f_n=x_n)$
provided that the series fulfills the so-called compatibility
condition, that is,
\begin{equation}
  \label{eq:compat}
  P_F^{(n)}(f_1=x_1,\ldots, f_n=x_n) =
  \sum_{x_{n+1}}P_F^{(n+1)}(f_1=x_1,\ldots, f_{n+1}=x_{n+1})
\end{equation}

\section{Syntax and Semantics}
\label{sec:semantics} 

Sato's distribution semantics, as summarized in Sec.~\ref{sec:ds},
provides the basis for most probabilistic logic programming languages
including PRISM~\cite{SatoKameya:01}, ICL~\cite{Poole08},
CP-logic~\cite{Vennekens09} and ProbLog~\cite{DeRaedt07-IJCAIa}.  The
precise way of defining the basic distribution $P_F$ differs among
languages, though the theoretical foundations are essentially the
same. The most basic instance of the distribution semantics, employed
by ProbLog, uses so-called \emph{probabilistic facts}.  Each ground
instance of a \emph{probabilistic fact} directly corresponds to an
independent random variable that takes either the value ``true'' or
``false''. These probabilistic facts can also be seen as binary
switches, cf.~\cite{Sato:95}, which again can be extended to multi-ary
switches or choices as used by PRISM and ICL. For switches, at most
one of the probabilistic facts belonging to the switch is ``true''
according to the specified distribution. Finally, in CP-logic, such
choices are used in the head of rules leading to the so-called
\emph{annotated disjunction}.

Hybrid ProbLog~\cite{gutmann10ilp} extends the distribution semantics
with continuous distributions.  To allow for exact inference, Hybrid
ProbLog imposes severe restrictions on the distributions and their
further use in the program.  Two sampled values, for instance, cannot
be compared against each other. Only comparisons that involve one
sampled value and one number constant are allowed. Sampled values may
not be used in arithmetic expressions or as parameters for other
distributions, for instance, it is not possible to sample a value and
use it as the mean of a Gaussian distribution. It is also not possible
to reason over an unknown number of objects as BLOG~\cite{Milch05}
does, though this is the case mainly for algorithmic reasons.

Here, we alleviate these restrictions by defining the basic
distribution $P_F$ over probabilistic facts based on both discrete and
continuous random variables.  We use a three-step approach to define
this distribution. First, we introduce explicit random variables and
corresponding distributions over their domains, both denoted by
terms. Second, we use a mapping from these terms to terms denoting
(sampled) outcomes, which, then, are used to define the basic
distribution $P_F$ on the level of probabilistic facts. For instance,
assume that an urn contains an unknown number of balls where the
number is drawn from a Poisson distribution and we say
that this urn contains many balls if it contains at least $10$
balls. We introduce a random variable $\mathtt{number}$, and we define
$\mathtt{many} \coloneq \mathtt{dist\_gt(\val(number), 9).}$ Here,
$\val(\mathtt{number})$ is the Herbrand term denoting the sampled
value of $\mathtt{number}$, and $\mathtt{dist\_gt(\val(number), 9)}$
is a probabilistic fact whose probability of being true is the
expectation that this value is actually greater than $9$. This
probability then carries over to the derived atom $\mathtt{many}$ as
well.  We will elaborate on the details in the following.

\subsection{Syntax}
\label{sec:syntax}

In a logic program, following Sato, we distinguish between
probabilistic facts, which are used to define the basic distribution,
and rules, which are used to derive additional atoms.\footnote{A rule
  can have an empty body, in which case it represents a deterministic
  fact.} Probabilistic facts are not allowed to unify with any rule
head.  The distribution over facts is based on random variables, whose
distributions we define through so called distributional clauses.

\begin{definition}[Distributional clause]
  \label{def:disclause}
  A \emph{distributional clause} is a definite clause with an atom
  $\mathtt{h} \sim \mathcal{D}$ in the head where $\sim$ is a binary
  predicate used in infix notation.
\end{definition}

For each ground instance $(\mathtt{h}\sim \mathcal{D} \coloneq
\mathtt{b_1,\ldots,b_n})\theta$ with $\theta$ being a substitution
over the Herbrand universe of the logic program, the distributional
clause defines a random variable $\mathtt{h}\theta$ and an associated
distribution $\mathcal{D}\theta$. In fact, the distribution is only
defined when $(\mathtt{b_1,\ldots,b_n})\theta$ is true in the
semantics of the logic program.  These random variables are terms of
the Herbrand universe and can be used as any other term in the logic
program.  Furthermore, a term $\val(d)$ constructed from the reserved
functor $\val/1$ represents the outcome of the random variable
$d$. These functors can be used inside calls to special predicates in
$dist\_rel =\{ dist\_eq/2, dist\_lt/2, dist\_leq/2, dist\_gt/2,
dist\_geq/2\}$.  We assume that there is a fact for each of the ground
instances of these predicate calls. These facts are the
\emph{probabilistic facts} of Sato's distribution semantics. Note that
the set of probabilistic facts is enumerable as the Herbrand universe
of the program is enumerable. A term $\val(d)$ links the random
variable $d$ with its outcome. The probabilistic facts compare the
outcome of a random variable with a constant or the outcome of another
random variable and succeed or fail according to the probability
distribution(s) of the random variable(s).

\begin{example}[Distributional clauses]
  \label{ex:probrules}

  \begin{align}
    \mathtt{nballs \sim poisson(6)}. 
         & \label{lbl:ex1}\\
    \mathtt{color(B)\sim\ [0.7:b, 0.3:g]} 
         &\coloneq \mathtt{between(1,\val(nballs),B).} \label{lbl:ex2}\\
    \mathtt{ diameter(B,MD)\sim gamma(MD/20,20) }
         &\mathtt{\coloneq  between(1,\val(nballs),B),}\nonumber \\
         &\quad\ \mathtt{mean\_diameter(\val(color(B)),MD).}\label{lbl:ex3}
  \end{align}

  The defined distributions depend on the following logical clauses:
  \begin{align*}
    \mathtt{mean\_diameter(C,5) } & \mathtt{\coloneq dist\_eq(C,b).}\\
    \mathtt{mean\_diameter(C,10) } & \mathtt{\coloneq dist\_eq(C,g). }\\
    \mathtt{ between(I,J,I)} & \mathtt{\coloneq dist\_leq(I,J).} \\
    \mathtt{between(I, J, K)} & \mathtt{\coloneq dist\_lt(I,J), \ I1\ is\ I + 1, between(I1,J,K)}.
  \end{align*}

  The distributional clause~\eqref{lbl:ex1} models the number of balls
  as a Poisson distribution with mean 6. The distributional
  clause~\eqref{lbl:ex2} models a discrete distribution for the random
  variable color(B). With probability 0.7 the ball is blue and green
  otherwise. Note that the distribution is defined only for the values
  B for which $between(1, \val(nballs), B)$ succeeds. Execution of
  calls to the latter give rise to calls to probabilistic facts that
  are instances of $dist\_leq(I,\val(nballs))$ and
  $dist\_lt(I,\val(nballs))$.  Similarly, the distributional
  clause~\eqref{lbl:ex3} defines a gamma distribution that is also
  conditionally defined. Note that the conditions in the distribution
  depend on calls of the form $mean\_diameter(\val(color(n)), MD)$
  with $n$ a value returned by between/3. Execution of this call
  finally leads to calls $dist\_eq(\val(color(n)),b)$ and
  $dist\_eq(\val(color(n)),g)$.
\end{example} 

It looks feasible, to allow $\val(d)$ terms everywhere and to have a
simple program analysis insert the special predicates in the
appropriate places by replacing $</2$, $>/2$, $\leq/2$, $\geq/2$
predicates by $dist\_rel/2$ facts.  Though extending unification is a
bit harder: as long as a $\val(h)$ term is unified with a free
variable, standard unification can be performed; only when the other
term is bound an extension is required.  In this paper, we assume that
the special predicates $\mathtt{dist\_eq/2}$, $\mathtt{dist\_lt/2}$,
$\mathtt{dist\_leq/2}$, $\mathtt{dist\_gt/2}$, and
$\mathtt{dist\_geq/2}$ are used whenever the outcome of a random
variable need to be compared with another value and that it is safe to
use standard unification whenever a $\mathtt{\val(h)}$ term is used in
another predicate.

For the basic distribution on facts to be well-defined, a program has
to fulfill a set of validity criteria that have to be enforced by the
programmer.

\begin{definition}[Valid program]
  \label{def:validprog}
  A program $P$ is called \emph{valid} if: 

 \begin{description}
  \item[(V1)] In the relation $\mathtt{h} \sim \mathcal{D}$ that holds
    in the least fixpoint of a program, there is a functional
    dependency from $\mathtt{h}$ to $ \mathcal{D}$, so there is a
    unique ground distribution $ \mathcal{D}$ for each ground random
    variable $\mathtt{h}$.

  \item[(V2)] The program is \emph{distribution-stratified}, that is,
    there exists a function $rank(\cdot)$ that maps ground atoms to
   $\mathbb{N}$ and that satisfies the following properties: (1)
    for each ground instance of a distribution clause $\mathtt{h}\sim
    \mathcal{D} \coloneq \mathtt{b_1, \ldots b_n}$ holds
    $rank(\mathtt{h}\sim \mathcal{D} > rank(\mathtt{b_i})$ (for all
    $i$).  (2) for each ground instance of another program clause:
    $\mathtt{h} \coloneq \mathtt{b_1, \ldots b_n}$ holds
    $rank(\mathtt{h})\geq rank(\mathtt{b_i})$ (for all $i$).  (3) for
    each ground atom $\mathtt{b}$ that contains (the name of) a random
    variable $\mathtt{h}$, $rank(\mathtt{b}) \geq rank(\mathtt{h} \sim
    \mathcal{D})$ (with $\mathtt{h} \sim \mathcal{D}$ the head of the
    distribution clause defining $\mathtt{h}$).

  \item[(V3)] All ground probabilistic facts or, to be more precise,
    the corresponding indicator functions are
    \emph{Lebesgue-measurable}.

  \item[(V4)] Each atom in the least fixpoint can be derived from a
    finite number of probabilistic facts (\emph{finite support
      condition}~\cite{Sato:95}).
  \end{description}
\end{definition}

Together, (V1) and (V2) ensure that a single basic distribution $P_F$
over the probabilistic facts can be obtained from the distributions of
individual random variables defined in $P$.  The requirement (V3) is
crucial.  It ensures that the series of distributions $P_F^{(n)}$
needed to construct this basic distribution is well-defined.  Finally,
the number of facts over which the basic distribution is defined needs
to be countable. This is true, as we have a finite number of constants
and functors: those appearing in the program.

\subsection{Distribution Semantics}
\label{sec:distributionsemantics}

We now define the series of distributions $P_F^{(n)}$ where we fix an
enumeration $f_1,f_2,\ldots$ of probabilistic facts such that $i < j
\implies rank(f_i) \leq rank(f_j)$ where $rank(\cdot)$ is a
\emph{ranking function} showing that the program is
distribution-stratified.  For each predicate $\mathtt{rel/2}\in
dist\_rel$, we define an \emph{indicator function} as follows:
\begin{align}
  I^1_{rel}(X_1,X_2) & =   
  \begin{cases}
    1  & \text{if } \mathtt{rel}(X_1,X_2) \text{ is true} \\
    0 & \text{if } \mathtt{rel}(X_1,X_2)  \text{ is false}
  \end{cases}
\end{align}

Furthermore, we set $I^0_{rel}(X_1,X_2) = 1.0 - I^1_{rel}(X_1,X_2)$.
We then use the expected value of the indicator function to define
probability distributions $P_F^{(n)}$ over finite sets of ground facts
$f_1,\ldots,f_n$. Let $\{rv_1,\ldots rv_m\}$ be the set of random
variables these $n$ facts depend on, ordered such that if
$rank(rv_i)<rank(rv_j)$, then $i<j$ (cf.~(V2) in
Definition~\ref{def:validprog}).  Furthermore, let $f_i =
rel_i(t_{i1},t_{i2})$, $x_j\in\{1,0\}$, and $\theta^{-1} =
\{\val(rv_1)/V_1,\ldots, \val(rv_m)/V_m\}$. The latter replaces all
evaluations of random variables on which the $f_i$ depend by variables
for integration.
\begin{align}
\lefteqn{P_F^{(n)}(f_1 = x_1, \ldots , f_n = x_n)
= \mathbb{E} [
I^{x_1}_{rel_1}(t_{11},t_{12}), \ldots ,
I^{x_n}_{rel_n}(t_{n1},t_{n2}) ]}\label{eq:series}\\
& = \idotsint \left( I^{x_1}_{rel_1}(t_{11}\theta^{-1}, t_{12}\theta^{-1})\cdots
I^{x_n}_{rel_n}(t_{n1}\theta^{-1}, t_{n2})\theta^{-1}\right) d\mathcal{D}_{rv_1}(V_1)\cdots d\mathcal{D}_{rv_m}(V_m)     \nonumber
\end{align}
\begin{example}[Basic Distribution]
Let $f_1,f_2,\ldots = dist\_lt(\val(b1), 3), dist\_lt(\val(b2), \val(b1)), \ldots$. The second distribution in the
series then is
\begin{align*}
\lefteqn{P_F^{(2)}(dist\_lt(\val(b1),  3)=x_1 ,dist\_lt(\val(b2), \val(b1))=x_2)}\\
  &= \mathbb{E} [ I_{\mathtt{dist\_lt}}^{x_1} (\val(b1),3),I_{\mathtt{dist\_lt}}^{x_2} (\val(b2),\val(b1))]\\
  &= \int\int \left(I_{\mathtt{dist\_lt} }^{x_1} (V1,3) ,I_{\mathtt{dist\_lt} }^{x_2} (V2,V1) \right)d\mathcal{D}_{b1}(V1)d\mathcal{D}_{b2}(V2)
\end{align*} 
\end{example}

By now we are able to prove the following proposition.

\begin{proposition}
  \label{prop:adm}
  Let $P$ be a valid program.  $P$ defines a probability measure $P_P$
  over the set of fixpoints of the $T_P$ operator.  Hence, $P$ also
  defines for an arbitrary formula $\mathtt{q}$ over atoms in its
  Herbrand base the probability that $\mathtt{q}$ is true.
\end{proposition}

\begin{proof}[Proof sketch] 
  It suffices to show that the series of distributions $P_F^{(n)}$
  over facts (cf.~\eqref{eq:series}) is of the form that is
  required in the distribution semantics, that is, these are
  well-defined probability distributions that satisfy the
  compatibility condition, cf.~\eqref{eq:compat}. This is a direct
  consequence of the definition in terms of indicator functions and
  the measurability of the underlying facts required for valid
  programs.
\end{proof}

\subsection{$T_P$ Semantics}
\label{sec:tp}

In the following, we give a procedural view onto the semantics by
extending the $T_P$ operator of Definition~\ref{def:tpop} to deal with
probabilistic facts $\mathtt{dist\_rel(t_1,t_2)}$. To do so, we
introduce a function \textsc{ReadTable}$(\cdot)$ that keeps track of
the sampled values of random variables to evaluate probabilistic
facts. This is required because interpretations of a program only
contain such probabilistic facts, but not the underlying outcomes of
random variables. Given a probabilistic fact
$\mathtt{dist\_rel(t1,t2)}$, \textsc{ReadTable} returns the truth
value of the fact based on the values of the random variables $h$ in
the arguments, which are either retrieved from the table or sampled
according to their definition $\mathtt{h}\sim \mathcal{D}$ as included
in the interpretation and stored in case they are not yet available.

\begin{definition}[Stochastic $T_P$ operator]
  \label{def:stp}
  Let $P$ be a valid program and $ground(P)$ the set of all ground
  instances of clauses in $P$.  Starting from a set of ground facts
  $S$ the $ST_P$ operator returns
  \begin{align*}
    ST_P(S) := \Big\{ \mathtt{h}\ \Big|\ &
       \mathtt{h \coloneq b_1,\ldots,b_n} \in ground(P) \mbox{ and }
       \forall\ \mathtt{b_i}: \mbox{ either } \mathtt{b_i}\in S \mbox{
         or }\\
       & \big(\mathtt{b_i}=dist\_rel(t1,t2) \wedge
       (t_j=\val(h)\rightarrow (\mathtt{h}\sim\mathcal{D})\in S) \wedge\\
    & \ \mbox{\textsc{ReadTable}}(b_i)=true  \big) \Big\}
\end{align*}
\end{definition}

\textsc{ReadTable} ensures that the basic facts are sampled from their
joint distribution as defined in Sec.~\ref{sec:distributionsemantics} during the
construction of the standard fixpoint of the logic program. Thus, each
fixpoint of the $ST_P$ operator corresponds to a possible
world whose probability is given by the distribution semantics.

\section{Forward sampling using Magic Sets and backward reasoning}
\label{sec:algorithms} 

In this section we introduce our new method for probabilistic forward
inference.  To this aim, we first extend the magic set transformation
to distributional clauses. We then develop a rejection sampling scheme
using this transformation. This scheme also incorporates backward
reasoning to check for consistency with the evidence during sampling
and thus to reduce the rejection rate.

\subsection{Probabilistic magic set transformation}
\label{sec:magic}

The disadvantage of forward reasoning in logic programming is that the
search is not goal-driven, which might generate irrelevant atoms.  The
\emph{magic set} transformation~\cite{bancilhon,nilsson:book} focuses
forward reasoning in logic programs towards a goal to avoid the
generation of uninteresting facts. It thus combines the advantages of
both reasoning directions.

\begin{definition}[Magic Set Transformation]
  \label{def:magicsets}
  If $P$ is a logic program, then we use $\textsc{Magic}(P)$ to denote the
  smallest program such that if $\mathtt{A_0 \coloneq A_1,\ldots, A_n}
  \in P$ then
  \begin{itemize}
  \item $\mathtt{A_0 \coloneq c(A_0), A_1, \ldots, A_n }\in \textsc{Magic}(P)$ and
  \item for each $1\le i \le n$: $\mathtt{c(A_i) \coloneq c(A_0),  A_1,\ldots, A_{i-1} }\in \textsc{Magic}(P)$
  \end{itemize}
\end{definition}

The meaning of the additional $\mathtt{c/1}$ atoms (c=call) is that
they ``switch on'' clauses when they are needed to prove a particular
goal. If the corresponding switch for the head atom is not true, the
body is not true and thus cannot be proven.  The magic transformation
is both sound and complete. Furthermore, if the SLD-tree of a goal is
finite, forward reasoning in the transformed program terminates. The
same holds if forward reasoning on the original program terminates.

We now extend this transformation to distributional clauses. The idea
is that the distributional clause for a random variable $h$ is
activated  when there is a call to a probabilistic fact
$dist\_rel(t_1, t_2)$ depending on $h$.

\begin{definition}[Probabilistic Magic Set Transformation]
  \label{def:pmagicsets}
  For program $P$, let $P_L$ be $P$ without distributional clauses.
  $\textsc{M}(P)$ is the smallest program s.t.  $\textsc{Magic}(P_L)
  \subseteq \textsc{M}(P)$ and for each $\mathtt{h} \sim \mathcal{D}
  \mathtt{ \coloneq b_1,\ldots,b_n} \in P$ and $\mathtt{rel} \in
  \{\mathtt{eq, lt, leq, gt, geq}\}$:
  \begin{itemize}
  \item $\mathtt{h} \sim \mathcal{D}\coloneq
    \mathtt{(c(dist\_rel(\val(h), X)); c(dist\_rel(X ,
      \val(h))),b_1,\ldots,b_n.} \in \textsc{M}(P)$.
  \item $\mathtt{c(b_i) \coloneq (c(dist\_rel(\val(h) , X));
      c(dist\_rel(X, \val(h))), b_1,\ldots, b_{i-1}. } \in \textsc{M}(P)$.
  \end{itemize}
   Then $\textsc{PMagic}(P)$ consists of:
  \begin{itemize}
  \item a clause $\mathtt{a\_p(t_1,\ldots,t_n)} \coloneq
    \mathtt{c(p(t_1,\ldots,t_n)), p(t_1,\ldots,t_n)}$ for each
    built-in predicate (including $\mathtt{dist\_rel/2}$ for
    $\mathtt{rel} \in \{\mathtt{eq, lt, leq, gt, geq} \}$) used in
    $\textsc{M}(P)$.

 \item a clause $\mathtt{h} \coloneq
    \mathtt{b_1',\ldots,b_n'}$ for each clause  $\mathtt{h} \coloneq \mathtt{b_1,\ldots,b_n} \in
    \textsc{M}(P)$ where $\mathtt{b_i'=a\_b_i}$ if $\mathtt{b_i}$ uses a built-in
    predicate and else $\mathtt{b_i'=b_i}$.
  \end{itemize}
\end{definition}

Note that every call to a built-in $\mathtt{b}$ is replaced by a call
to $\mathtt{a\_b}$; the latter predicate is defined by a clause that
is activated when there is a call to the built-in ($\mathtt{c(b)}$)
and that effectively calls the built-in.  The transformed program
computes the distributions only for random variables whose value is
relevant to the query. These distributions are the same as those
obtained in a forward computation of the original program. Hence we
can show:

\begin{lemma}
 Let $P$ be a program and $\textsc{PMagic}(P)$ its probabilistic
  magic set transformation extended with a seed $c(q)$.  The
  distribution over $q$ defined by $P$ and by $\textsc{PMagic}(P)$ is
  the same.
\end{lemma}

\begin{proof}[Proof sketch]
  In both programs, the distribution over $q$ is determined by the
  distributions of the atoms $dist\_eq(t_1,t_2)$,
  $dist\_leq(t_1,t_2)$, $dist\_lt(t_1,t_2)$, $dist\_geq(t_1,t_2)$, and
  $dist\_gt(t_1,t_2)$ on which $q$ depends in a forward computation of
  the program $P$. The magic set transformation ensures that these
  atoms are called in the forward execution of $\textsc{PMagic}(P)$.
  In $\textsc{PMagic}(P)$, a call to such an atom activates the
  distributional clause for the involved random variable. As this
  distributional clause is a logic program clause, soundness and
  completeness of the magic set transformation ensures that the
  distribution obtained for that random variable is the same as in
  $P$.  Hence also the distribution over $q$ is the same for both
  programs.
\end{proof}

\subsection{Rejection sampling with heuristic lookahead}

As discussed in Section~\ref{sec:probinf}, sampling-based approaches
to probabilistic inference estimate the conditional probability
$p(q|e)$ of a query $q$ given evidence $e$ by randomly generating a
large number of samples or possible
worlds~(cf.~Algorithm~\ref{alg:sampling}). The algorithm starts by
preparing the program $L$ for sampling by applying the \textsc{PMagic}
transformation.  In the following, we discuss our choice of subroutine
\textsc{STPMagic} (cf.~Algorithm~\ref{alg:rejection}) which realizes
likelihood weighted sampling.  It is used in
Algorithm~\ref{alg:sampling}, line~\ref{line:callWeightedsample}, to
generate individual samples. It iterates the stochastic consequence
operator of Definition~\ref{def:stp} until either a fixpoint is
reached or the current sample is inconsistent with the
evidence. Finally, if the sample is inconsistent with the evidence, it
receives weight 0.

\begin{algorithm}[t]
  \caption{Main loop for sampling-based inference to calculate the
    conditional probability $p(q|e)$ for query $q$, evidence $e$ and
    program $L$. }
  \label{alg:sampling}
\begin{algorithmic}[1]
\Function{\textsc{Evaluate}}{$L$, $q$, $e$, $Depth$}
\State $L^*:=$\textsc{PMagic}$(L)\cup \{c(a) | a \in e \cup{q} \}$
\State $n^+ :=0$ \hspace{1cm} $n^- := 0$
\While{Not converged}
\State $(I,w):=$\textsc{STPMagic}$(L^*,L,e,Depth)$\label{line:callWeightedsample}
\State \textbf{if} $q\in I$ \textbf{then}  $n^+:= n^+ + w$
\textbf{else}  $n^-:= n^- + w$ 
\EndWhile
\State \Return $n^+ / (n^+ + n^-)$
\EndFunction
\end{algorithmic}
\end{algorithm}

\begin{algorithm}[t]
  \caption{Sampling one interpretation; used in
    Algorithm~\ref{alg:sampling}. }
  \label{alg:rejection}
  \begin{algorithmic}[1]
    \Function{STPMagic}{$L^*,L,e,Depth$}
    \State $T_{pf}:=\varnothing$, $T_{dis}:=\varnothing$, $w:=1$, $I_{old}:=\varnothing$, $I_{new}:=\varnothing$
    \Repeat 
\State $I_{old}:=I_{new}$
\ForAll{$(\mathtt{h \coloneq \
        body)}\ \in L^*$}
\State split body in $B_{PF}$ (prob. facts) and $B_L$ (the rest)
\ForAll{grounding substitution $\theta$ such that $B_L\theta \subseteq I_{old}$}
\State $s := true$, $w_d := 1$
\While {$s \wedge B_{PF}\neq \varnothing$}
\State select and remove $pf$ from $B_{PF}$
\State $(b_{pf},w_{pf}) := $\textsc{ReadTable}$(pf\theta, I_{old},
T_{pf}, T_{dis}, L, e, Depth)$ \label{line:samplehead} 
\State $s := s  \wedge b_{pf}$ \hspace{1cm}  $w_d := w_d \cdot w_{pf}$
\EndWhile
\If {$s$} \If {$h\theta \in e^-$} \Return $(I_{new}, 0)$
\Comment{check negative evidence }\EndIf
\State $I_{new} := I_{new} \cup \{h\theta\}$ \hspace{1cm} $w := w \cdot w_d$
\EndIf
\EndFor
\EndFor
     \Until {$I_{new}=I_{old} \vee w=0$} \Comment{Fixpoint or
      impossible evidence}
\If {$e^+ \subseteq I_{new}$} \Return $(I_{new}, w)$ \Comment{check
  positive evidence} \Else \ \Return $(I_{new}, 0)$\EndIf
       \EndFunction
\end{algorithmic}
\end{algorithm}

\begin{algorithm}[ht]
  \caption{Evaluating a probabilistic fact $pf$; used in
    Algorithm~\ref{alg:rejection}. \textsc{ComputePF}$(pf,T_{dis})$
    computes the truth value and the probability of $pf$ according to
    the information in $T_{dis}$.}
  \label{alg:likelihoodweighting}
  \begin{algorithmic}[1]
    \Function{ReadTable}{$pf, I,T_{pf}, T_{dis}, L, e,Depth$}
\If{$pf \notin T_{pf}$}
\ForAll {random variable $h$ occurring in $pf$ where $h \notin T_{dis}$}
\If{$h \sim D \notin I $}  \Return $(false,0)$  \EndIf
\If {not \Call{Sample}{$h, D,  T_{dis}, I, L, e, Depth$}}  \Return
$(false,0)$ \label{line:sample} \EndIf
\EndFor
\State $(b,w) :=$ \Call{ComputePF}{$pf$,$T_{dis}$} 
\If {$(b \wedge (pf \in e^-)) \vee (\neg b \wedge (pf \in e^+)) $}\State\Return $(false,0)$ \Comment{inconsistent with evidence}\EndIf
\State extend $T_{pf}$ with $(pf,b,w)$
\EndIf
\State\Return $(b,w)$ as stored in $T_{pf}$ for $pf$
    \EndFunction
  \Procedure{Sample}{$\mathtt{h}, \mathcal{D},\ T_{dis}, \ I, \ L,\ e,\
    Depth$}
  \State $w_h:=1$, $\mathcal{D}' := \mathcal{D}$   \Comment{Initial
    weight, temp. distribution}
  \If{$\mathcal{D}' = [p_1:\mathtt{a_1} ,\ldots , p_n:\mathtt{a_n}]$} \Comment{finite distribution}
 \For{$p_j:\mathtt{a_j} \in \mathcal{D}'$ where $\mathtt{dist\_eq}(\mathtt{h},\mathtt{a_j}) \in e^-$} \Comment{remove neg.~evidence}
  \State $\mathcal{D}' := \Call{Norm}{\mathcal{D}' \setminus
  \{p_j:\mathtt{a_j}\}}$, \hspace{1cm}$w_h:=w_h\times (1-p_j)$
 \EndFor
\If {$\exists v: \mathtt{dist\_eq(\val(h),v)} \in e^+ $ and $p:\mathtt{v} \in \mathcal{D}'$}
\State $\mathcal{D}' := [1 : v]$, \hspace{1cm} $w_h:=w_h\times
p$ \EndIf

 \For{$p_j:\mathtt{a_j} \in \mathcal{D}'$} \Comment{remove choices that make
   $e^+$ impossible}
 \If{$\exists b\in e^+$: not \Call{MaybeProof}{$b,Depth , I \cup
     \lbrace \mathtt{dist\_eq}(\mathtt{h},\mathtt{a_j}) \rbrace,L$}  \textbf{or}\\
 \hspace{1.8cm} $\exists b\in e^-$: not \Call{MaybeFail}{$b,Depth , I \cup
     \lbrace \mathtt{dist\_eq}(\mathtt{h},\mathtt{a_j}) \rbrace,L$}
}
  \State $\mathcal{D}' := \Call{Norm}{\mathcal{D}' \setminus
  \{p_j:\mathtt{a_j}\}}$, \hspace{1cm}$w_h:=w_h\times (1-p_j)$
\EndIf
 \EndFor
\EndIf
 \State \textbf{if} {$\mathcal{D}'=\varnothing$} \Return false
 \State Sample $x$ according to $\mathcal{D'}$, extend $T_{dis}$ with $(h,x)$ and \Return true
\EndProcedure
  \end{algorithmic}
\end{algorithm}

Algorithm~\ref{alg:likelihoodweighting} details the procedure used in
line~\ref{line:samplehead} of Algorithm~\ref{alg:rejection} to sample
from a given distributional clause. The function \textsc{ReadTable}
returns the truth value of the probabilistic fact, together with its
weight. If the outcome is not yet tabled, it is computed. Note that
\texttt{false} is returned when the outcome is not consistent with the
evidence. Involved distributions, if not yet tabled, are sampled in
line~\ref{line:sample}. In the infinite case, \textsc{Sample} simply
returns the sampled value. In the finite case, it is directed towards
generating samples that are consistent with the evidence. Firstly, all
possible choices that are inconsistent with the negative evidence are
removed. Secondly, when there is positive evidence for a particular
value, only that value is left in the distribution. Thirdly, it is
checked whether each left value is consistent with all other
evidence. This consistency check is performed by a simple
depth-bounded meta-interpreter. For positive evidence, it attempts a
top-down proof of the evidence atom in the original program using the
function \textsc{MaybeProof}. Subgoals for which the depth-bound is
reached, as well as probabilistic facts that are not yet tabled are
assumed to succeed. If this results in a proof, the value is
consistent, otherwise it is removed. Similarly for negative evidence:
in \textsc{MaybeFail}, subgoals for which the depth-bound is reached,
as well as probabilistic facts that are not yet tabled are assumed to
fail. If this results in failure, the value is consistent, otherwise
it is removed. The $Depth$ parameter allows one to trade the
computational cost associated with this consistency check for a
reduced rejection rate.

Note that the modified distribution is normalized and the weight is
adjusted in each of these three cases. The weight adjustment takes
into account that removed elements cannot be sampled and is necessary
as it can depend on the distributions sampled so far which elements
are removed from the distribution sampled in \textsc{Sample} (the
clause bodies of the distribution clause are instantiating the
distribution).

\section{Experiments}
\label{sec:experiments}

We implemented our algorithm in YAP Prolog and set up experiments to
answer the questions
\begin{itemize}
\item[\textbf{Q1}] Does the lookahead-based sampling improve the
  performance?
\item[\textbf{Q2}] How do rejection sampling and likelihood weighting compare?
\end{itemize}

To answer the first question, we used the distributional program in
Figure~\ref{fig:experimentnogreen}, which models an urn containing a
random number of balls. The number of balls is uniformly distributed
between 1 and 10 and each ball is either red or green with equal
probability. We draw 8 times a ball with replacement from the urn and
observe its color. We also define the atom
$\mathtt{nogreen(}D\mathtt{)}$ to be true if and only if we did not
draw any green ball in draw 1 to $D$.  We evaluated the query
$P(\mathtt{dist\_eq(\simeq(color(\simeq(drawnball(1)))),red)}\
|\mathtt{nogreen(}D\mathtt{)})$ for $D=1,2,\ldots, 8$.  Note that the
evidence implies that the first drawn ball is red, hence that the
probability of the query is 1; however, the number of steps required
to proof that the evidence is inconsistent with drawing a green first
ball increases with D, so the larger is D, the larger Depth is
required to reach a 100\% acceptance rate for the sample as
illustrated in Figure ~\ref{fig:experimentnogreen}.  It is clear that
by increasing the depth limit, each sample will take longer to be
generated. Thus, the $Depth$ parameter allows one to trade off
convergence speed of the sampling and the time each sample needs to be
generated. Depending on the program, the query, and the evidence there
is an optimal depth for the lookahead.

\begin{figure}[ht]
  \centering
  \begin{minipage}{0.40\textwidth}
    \begin{align*}
      &{ }\\
      &{ }\\
      &{ }\\
      &\mathtt{numballs \sim uniform([1,2,3,4,5,6,7,8,9,10]).}\\
      &\mathtt{ball(M) \coloneq between(1,numballs,M).}\\
      &\mathtt{color(B)\sim uniform([red,green]) \coloneq ball(B).}\\
      &\mathtt{draw(N) \coloneq between(1,8,N).}\\
      &\mathtt{nogreen(0).}\\
    \end{align*}
  \end{minipage}
  \begin{minipage}{0.42\textwidth}
    \includegraphics[scale=0.44]{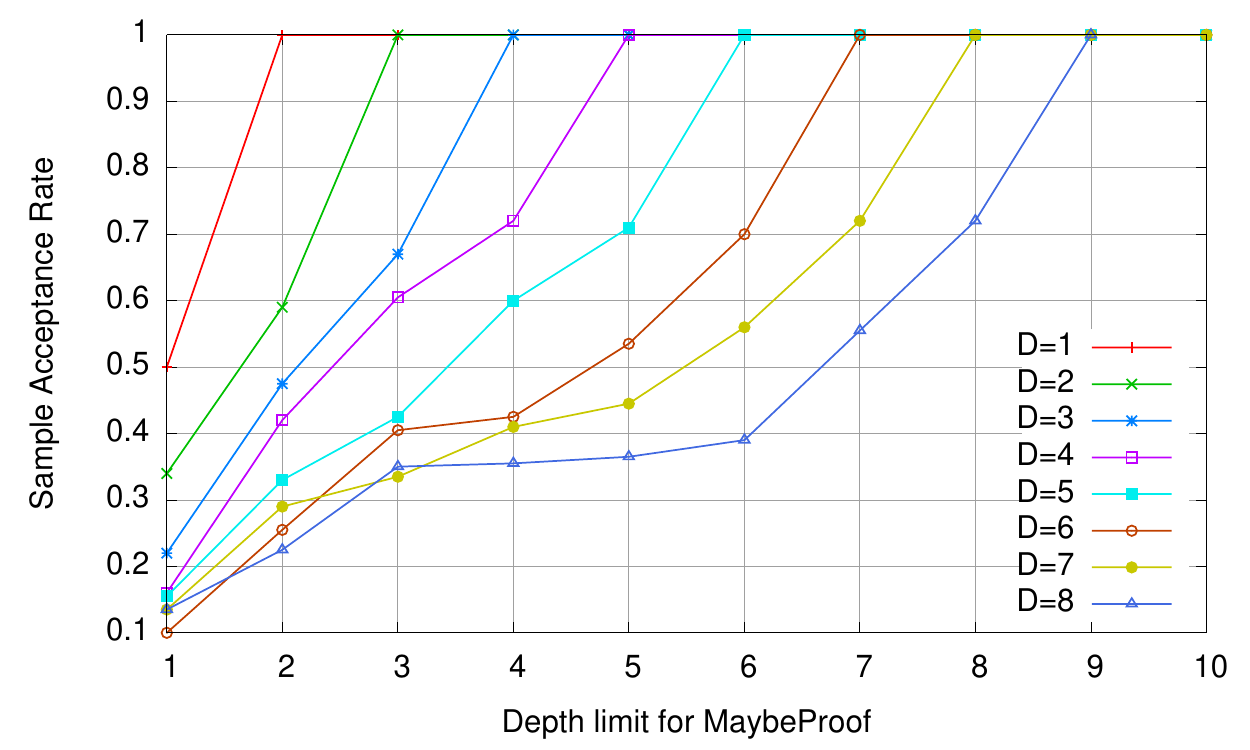}
  \end{minipage}
  \vspace{-0.5cm}
  \begin{align*}
    &\mathtt{nogreen(D) \coloneq
      dist\_eq(\simeq(color(\simeq(drawnball(D)))),red),\ D2\ is\
      D-1,\
      nogreen(D2).}\\
    &\mathtt{drawnball(D) \sim uniform(L) \coloneq
      draw(D),}\mathtt{findall(B,ball(B),L).}
  \end{align*}
  \caption{ The program modeling the urn (left); rate of accepted
    samples (right) for evaluating the query
    $P(\mathtt{dist\_eq(\simeq(color(\simeq(drawnball(1)))),red)}\ |$
    $\mathtt{nogreen(}D\mathtt{)})$ for $D=1,2,\ldots, 10$ and for
    $Depth=1,2,\ldots,8$ using Algorithm~\ref{alg:sampling}. The
    acceptance rate is calculated by generating 200 samples using our
    algorithm and counting the number of sample with weight larger
    than 0.}
  \label{fig:experimentnogreen}
\end{figure}

To answer \textbf{Q2}, we used the standard example for
BLOG~\cite{Milch05}.  An urn contains an unknown number of balls where
every ball can be either green or blue with $p=0.5$. When drawing a
ball from the urn, we observe its color but do not know which ball it
is.  When we observe the color of a particular ball, there is a $20\%$
chance to observe the wrong one, e.g. green instead of blue. We have
some prior belief over the number of balls in the urn. If 10 balls are
drawn with replacement from the urn and we saw 10 times the color
green, what is the probability that there are $n$ balls in the urn? We
consider two different prior distributions: in the first case, the
number of balls is uniformly distributed between 1 and 8, in the
second case, it is Poisson-distributed with mean $\lambda=6$.

\begin{figure}[t]
  \centering
  \includegraphics[scale=0.49]{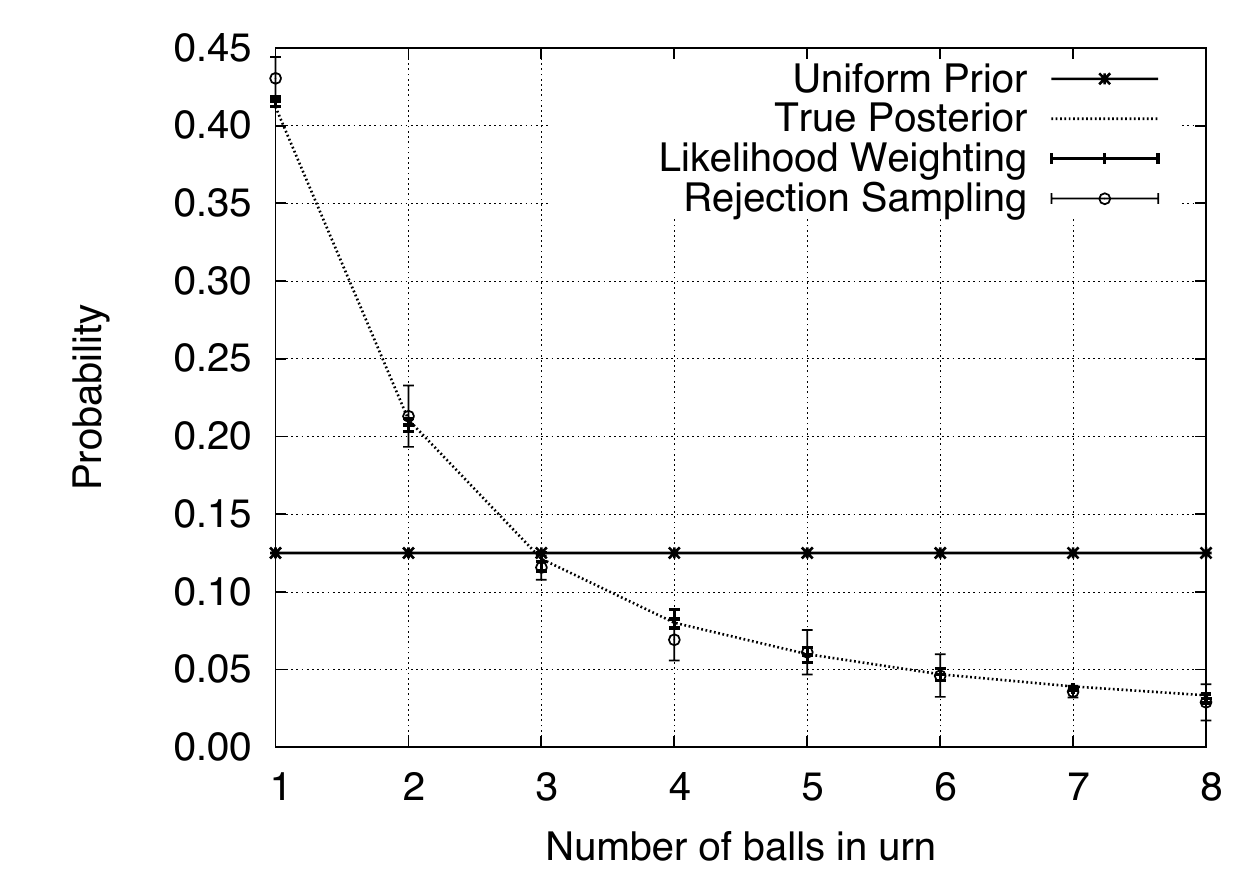}
  \includegraphics[scale=0.49]{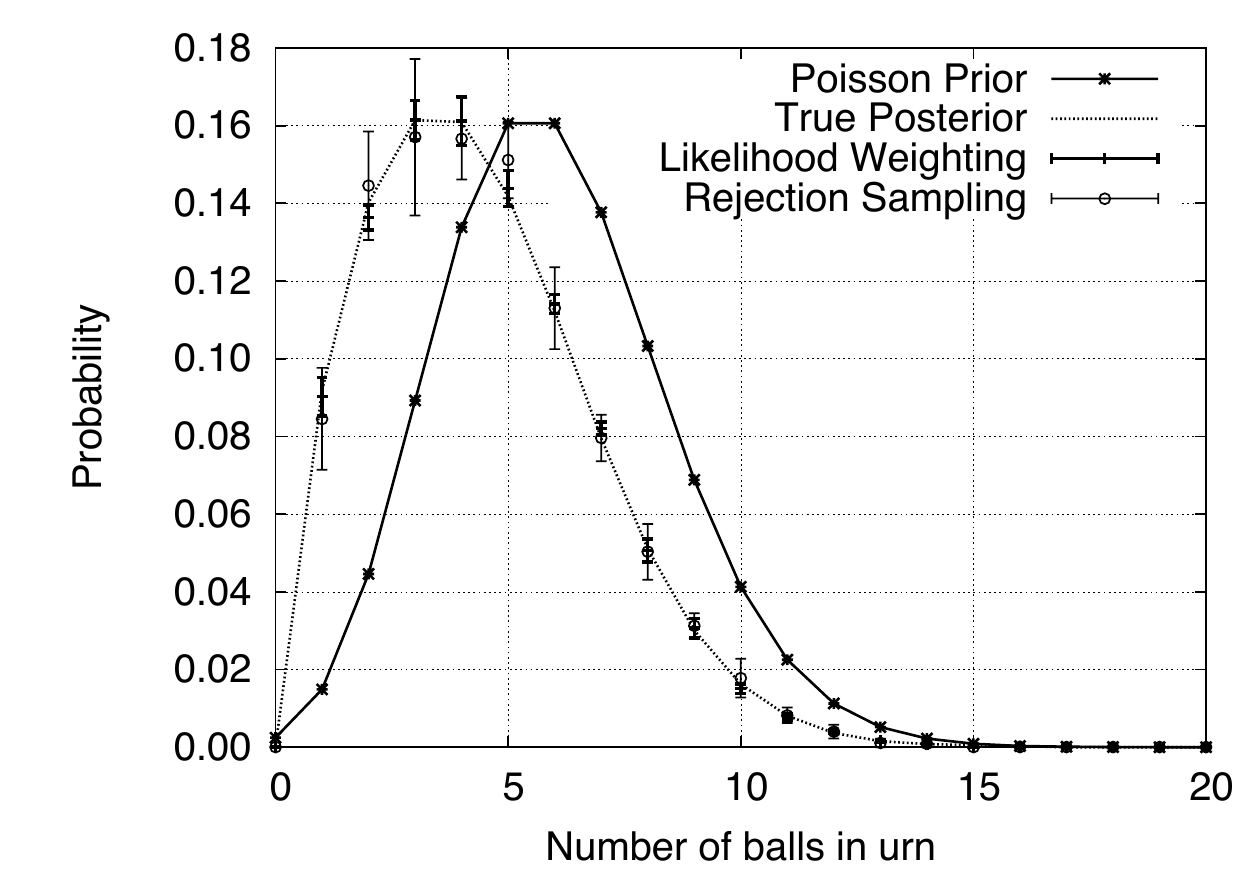}
  \caption{Results of urn experiment with forward reasoning. 10 balls
    with replacement were drawn and each time green was
    observed. Left: Uniform prior over \# balls, right: Poisson prior
    $(\lambda=6)$.}
  \label{fig:urnproblog}
\end{figure}

One run of the experiment corresponds to sampling the number $N$ of
balls in the urn, the color for each of the $N$ balls, and for each of
the ten draws both the ball drawn and whether or not the color is
observed correctly in this draw. Once these values are fixed, the
sequence of colors observed is determined. This implies that for a
fixed number $N$ of balls, there are $2^N\cdot N^{10}$ possible
proofs.  In case of the uniform distribution, exact PRISM inference
can be used to calculate the probability for each given number of
balls, with a total runtime of $0.16$ seconds for all eight cases. In
the case of the Poisson distribution, this is only possible up to 13
balls, with more balls, PRISM runs out of memory.  For inference using
sampling, we generate 20,000 samples with the uniform prior, and
100,000 with Poisson prior. We report average results over five
repetitions.  For these priors, PRISM generates 8,015 and 7,507
samples per second respectively, ProbLog backward sampling 708 and
510, BLOG 3,008 and 2,900, and our new forward sampling (with
rejection sampling) 760 and 731.  The results using our algorithm for
both rejection sampling and likelihood weighting with $Depth=0$ are
shown in Figure~\ref{fig:urnproblog}. As the graphs show, the standard
deviation for rejection sampling is much larger than for likelihood weighting.

\section{Conclusions and related work}
\label{sec:conclusions}

We have contributed a novel construct for probabilistic logic
programming, the distributional clauses, and defined its semantics.
Distributional clauses allow one to represent continuous variables and
to reason about an unknown number of objects. In this regard this
construct is similar in spirit to languages such as BLOG and Church,
but it is strongly embedded in a logic programming context.  This
embedding allowed us to propose also a novel inference method based on
a combination of importance sampling and forward reasoning. This
contrasts with the majority of probabilistic logic programming
languages which are based on backward reasoning (possibly enhanced
with tabling \cite{SatoKameya:01,Mantadelis10iclp}).  Furthermore,
only few of these techniques employ sampling, but see \cite{Kimmig11}
for a Monte Carlo approach using backward reasoning.  Another key
difference with the existing probabilistic logic programming
approaches is that the described inference method can handle
evidence. This is due to the magic set transformation that targets the
generative process towards the query and evidence and instantiates
only relevant random variables.  

P-log~\cite{Baral09} is a probabilistic language based on Answer Set
Prolog (ASP). It uses a standard ASP solver for inference and thus is
based on forward reasoning, but without the use of sampling.  Magic
sets are also used in probabilistic Datalog~\cite{Fuhr00}, as well as
in Dyna, a probabilistic logic programming language \cite{Eisner05}
based on rewrite rules that uses forward reasoning. However, neither
of them uses sampling. Furthermore, Dyna and PRISM require that the
exclusive-explanation assumption.  This assumption states that no two
different proofs for the same goal can be true simultaneously, that
is, they have to rely on at least one basic random variable with
different outcome. Distributional clauses (and the ProbLog language)
do not impose such a restriction.  Other related work includes
MCMC-based sampling algorithms such as the approach for
SLP~\cite{cussensmc}.  Church's inference algorithm is based on MCMC
too, and also BLOG is able to employ MCMC.  At least for BLOG it seems
to be required to define a domain-specific proposal distribution for
fast convergence.  With regard to future work, it would be interesting
to consider evidence on continuous distributions as it is currently
restricted to finite distribution. Program analysis and transformation
techniques to further optimize the program w.r.t. the evidence and
query could be used to increase the sampling speed.  Finally, the
implementation could be optimized by memoizing some information from
previous runs and then use it to more rapidly prune as well as sample.

\section*{Acknowledgements} 

Angelika Kimmig and Bernd Gutmann are supported by the Research
Foundation-Flanders (FWO-Vlaanderen). This work is supported by the
GOA project 2008/08 Probabilistic Logic Learning and by the European
Community's Seventh Framework Programme under grant agreement
First-MM-248258.

\newpage 
 
\bibliographystyle{acmtrans}
\bibliography{bibtex}

\end{document}